\documentclass[a4paper,UKenglish]{lipics-v2016}

\usepackage{microtype}

\usepackage[dvips]{epsfig}
\usepackage{amssymb, amsmath, amsthm}
\usepackage{multicol}
\usepackage{algorithm}
\usepackage[noend]{algorithmic}

\newcommand{\WHEN}{\WHILE}
\newcommand{\ENDWHEN}{\ENDWHILE}
\floatname{algorithm}{Protocol}

\theoremstyle{plain}

\newtheorem{proposition}[theorem]{Proposition}

\newcommand{\ceil}[1]{\left\lceil#1\right\rceil}
\newcommand{\floor}[1]{\left\lfloor#1\right\rfloor}

\newcommand{\conceptFormat}[1]{\textbf{#1}}
\newcommand{\concept}[1]{\index{#1}\conceptFormat{#1}}

\def\size{\textit{c}}
\def\sizeT{\textit{c}}
\def\phase{\textit{phase}}
\def\cnt{\textit{cnt}}
\def\U{\textsl{U}}

\def\Agents { \mathcal{A}}
\def\n{n}

\newcommand{\INDSTATE}[1][1]{\STATE\hspace{#1\algorithmicindent}}
\title{Time and Space Optimal Counting in Population Protocols} 
\titlerunning{Time and Space Optimal Counting in Population Protocols} 

\author[1]{James Aspnes\thanks{The work of this author is supported in part by NSF grants CCF-1650596 and CCF-1637385.}}
\author[2]{Joffroy Beauquier}
\author[2]{Janna Burman\thanks{contact author}}
\author[3]{Devan Sohier}
\affil[1]{Yale University, USA\\
  \texttt{james.aspnes@gmail.com}}
\affil[2]{Universit\'{e} Paris Sud, LRI, France\\
  \texttt{\{joffroy.beauquier, janna.burman\}@lri.fr}}
\affil[3]{Universit\'{e} de Versailles, LI-PaRAD, France\\
  \texttt{devan.sohier@uvsq.fr}}
\authorrunning{J. Aspnes, J. Beauquier, J. Burman and D. Sohier} 

\Copyright{James Aspnes, Joffroy Beauquier, Janna Burman and Devan Sohier}


\begin{document}

\maketitle

\begin{abstract}
This work concerns the general issue of combined optimality in terms of time and space complexity. In this context, we study the problem of (exact) \textit{counting} resource-limited and passively mobile nodes in the model of \textit{population protocols}, in which the space complexity is crucial. The counted nodes are memory-limited anonymous devices (called agents) communicating asynchronously in pairs (according to a \textit{fairness} condition).
Moreover, we assume that these agents are prone to failures so that they cannot be correctly initialized.

%

This study considers two classical fairness conditions, and for each we investigate the issue of time optimality of counting given the optimal space per agent. In the case of randomly interacting agents (\textit{probabilistic fairness}), as usual, the convergence time is measured in terms of \textit{parallel time} (or parallel interactions), which is defined as the number of pairwise interactions until convergence, divided by $\n$ (the number of agents). In case of \textit{weak fairness}, where it is only required that every pair of agents interacts infinitely often, the convergence time is defined in terms of \textit{non-null transitions}, i.e, the transitions that affect the states of the interacting agents.

%

First, assuming probabilistic fairness, we present a ``non-guessing'' time optimal protocol of $O(\n \log \n)$ expected time given an optimal space of only one bit, and we prove the time optimality of this protocol.
Then, for weak fairness, we show that a space optimal (\textit{semi-uniform}) solution cannot converge faster than in $\Omega(2^{\n})$ time (non-null transitions). This result, together with the time complexity analysis of an already known space optimal protocol, shows that it is also optimal in time (given the optimal space constrains).

 \end{abstract}

\section{Introduction}
\setcounter{footnote}{0}


In this paper we are interested in
the determination of the \emph{exact} number of nodes in a mobile sensor network. 
In the considered networks, sensors are typically attached to mobile supports (vehicles, animals, people, etc.) moving in an unpredictable way. Moreover, nodes may be deployed at large scale. Therefore they should be cheap and consequently can be prone to failures. One can think of sensors attached to 
zebras (ZebraNet \cite{zebra}), pigeons (Pigeon Air Patrol \cite{pigeonairpatrol_website}), or public transport vehicles (EMMA project \cite{DBLP:journals/wicomm/LahdeDPLW07}). In this context, counting can be part of the task being realized (How many animals have a temperature exceeding some bound?), or part of the managing of the network (Should some nodes be added or replaced?). In relation with the domain of application we are looking at, we consider that the nodes are anonymous, undistinguishable, have a bounded memory and poor communication capabilities (no broadcast; only a pairwise communication when two nodes come close enough to each other).
A distributed computing model that fits this description is the model of \emph{population protocols} (PP) \cite{DBLP:journals/dc/AngluinADFP06}.

In PP, mobile nodes, which are called agents, can be represented as finite state transition systems. One can imagine that, when two agents are close enough, they interact and the effect of the interaction is a transition with a possible change of states. In this work we study the case of \emph{symmetric protocols}, where two agents in a transition are indistinguishable if their states are identical (thus, their states have to be identical also after the transition). This assumption makes the protocol design more difficult than in the asymmetric case (as it restrains the set of possible transition rules), but provides a more general solution (correct in both cases).

The mobility and the resulting interactions of agents are completely asynchronous and modeled in a very general way - by a \emph{fairness} assumption. Here, we study the problem of counting considering two classical fairness assumptions. One ensures that each pair of agents is drawn uniformly at random for each interaction, and the other, weaker assumption (called here \emph{weak}), ensures only that every pair of agents interact infinitely often. While the probabilistic fairness captures the randomization inherent to many real systems, weak fairness only ensures progress of system entities (see Sec. \ref{sec:model} for an illustrating example). 

As the agents are likely to be cheap and prone to failures (memory corruption, crash failures, etc.), re-counting may be required frequently. Since the population of agents may be very large, re-initialization may be infeasible. Hence, it is natural to assume that the agents are not initialized (i.e., an agent can be initially in any possible state). However, it is easy to prove that, in this case, counting in PP is impossible \cite{Beauquier2007}. 
The solution is to use only 
one initialized and distinguishable agent, called the \emph{base station} (BST). In this work, for the first time, we also prove, the necessity of such an agent being distinguishable (in case of symmetric protocols; see Sec. \ref{sec:lower-bound-weak}). BST is also the agent that eventually obtains the correct count of the population. Thus the considered protocols are \emph{semi-uniform}, in the sense that all the agents, except BST, are a priori undistinguishable and execute the same protocol, for any population size $\n$ and upper bound $P$ on $\n$ (see Sec. \ref{sec:model} for a formal definition). 

In this context, previous works \cite{Beauquier2007,DBLP:journals/tcs/IzumiKIW14,beauquier:hal-00986109} and a companion paper \cite{DBLP:conf/wdag/BeauquierBCS15} study the issue of space complexity (of the counted nodes), a factor that is particulary important in the considered large-scale and unreliable networks.
For instance,  \cite{Beauquier2007} shows that under weak fairness, 
$P$ (or more) agents cannot be counted with strictly less than $P$ states per counted agent, by deterministic protocols (considered here as well). Here, as a by-product, we present an alternative proof of this result, for the case of symmetric protocols (see Proposition \ref{prop-count-weak-Ps}).
However, as shown in \cite{DBLP:conf/wdag/BeauquierBCS15}, under probabilistic or \emph{global fairness}\footnote{Global fairness can be viewed as simulating probabilistic one without introducing randomization explicitly. One can see probabilistic fairness as a quantitative version of the global one. Moreover, it allows to analyze protocols' time complexity, what is impossible in general with global fairness (see Sec.~\ref{sec:model}).}, counting can be performed with only two states (one bit) per counted agent.
\cite{DBLP:conf/wdag/BeauquierBCS15} presents two space optimal solutions to counting in PP, one under weak fairness and the other under global fairness (the latter is also correct under probabilistic fairness). The solution for global fairness uses a memory of only one bit per agent, while the solution for weak fairness needs $\log P$ bits ($P$ states) per counted agent.
In the latter case, as the number $P$ of states for a given counting protocol is fixed, it represents an upper bound on the size of populations to which the protocol applies, and is considered as an explicit parameter.

In addition to the state space optimality, this paper raises the issue of the convergence time. Our objective is to determine the best guarantees on the convergence times given the established necessary minimum on the memory. 
To obtain this goal, we show, in particular, that the convergence times of the two previous space optimal solutions (in \cite{DBLP:conf/wdag/BeauquierBCS15}) are exponential. 
In Sec. \ref{sec:probabilistic}, for reasons explained there, we restrict our attention to so called ``non-guessing'' counting protocols.
We prove that any such state optimal counting protocol, correct under probabilistic fairness, converges in $\Omega(n \log n)$ expected parallel interactions (Sec. \ref{sec:lower-bound-probabilistic}). In Sec. \ref{sec:time-space-opt-prot}, we propose a new state optimal protocol fitting this complexity.

In the case of weak fairness, in Sec. \ref{sec:lower-bound-weak}, we show that a space optimal solution requires an exponential convergence time (in terms of non-null transitions).
In particular, this result shows that the space optimal protocol under weak fairness in \cite{DBLP:conf/wdag/BeauquierBCS15} is time optimal among the space optimal semi-uniform protocols.


 \paragraph*{Related Work.} We provide here the most relevant and recent works. Please, refer to \cite{DBLP:conf/wdag/BeauquierBCS15} for additional related literature on the subject.

Considering PP, \cite{Beauquier2007,DBLP:journals/tcs/IzumiKIW14,beauquier:hal-00986109} proposed efficient counting protocols in terms of exact state space that were improved in \cite{DBLP:conf/wdag/BeauquierBCS15} by space optimal solutions. For weak fairness, the protocol proposed in \cite{DBLP:journals/tcs/IzumiKIW14} uses $2P$ states per agent (only 1 bit more than the optimal) and converges in only $O(\log n)$ asynchronous rounds (a round being a shortest fragment of execution where each agent interacts with each other). This presents an interesting trade-off for counting in PP. A recent work \cite{DBLP:conf/nca/MocquardAABS15} studies a problem related to counting in random PP, where agents should determine the difference between the number of agents started in state $A$ and the number of those started in state $B$. In contrast with the current work, \cite{DBLP:conf/nca/MocquardAABS15} assumes initialized agents, but similarly to the current work, it investigates the efficiency in terms of both time and space. It presents an $O(n^{3/2})$-state population protocol that allows each agent to converge to the exact solution by interacting no more than $O(\log n)$ times. Additional very recent works (as \cite{DBLP:journals/corr/AlistarhAEGR16,DBLP:conf/wdag/DotyS15,DBLP:conf/icalp/AlistarhG15,DBLP:conf/podc/AlistarhGV15}) jointly contribute to the time and space trade-offs study of fundamental tasks, as majority and leader election, in random PP. For example, \cite{DBLP:conf/wdag/DotyS15} shows that it is impossible to achieve sub-linear leader election with only constant state space per agent, but due to \cite{DBLP:conf/icalp/AlistarhG15} this problem can be solved in $O(\log^3 n)$ time with $O(\log^3 n)$ states. For majority, sub-linear time is impossible for protocols with at most four states per node, while there exists a poly-logarithmic time protocol which requires a linear in $n$ state space  \cite{DBLP:conf/podc/AlistarhGV15}. \cite{DBLP:journals/corr/AlistarhAEGR16} presents additional upper and lower bounds for these tasks that  
highlight a time complexity separation between $O(\log \log n)$ and $\Theta(\log^2 n)$ state space for both majority and leader election. The present work contributes to the general study of time-space trade-offs in the case of counting.

In the context of dynamic networks with anonymous nodes, recent works \cite{DBLP:conf/sss/MichailCS13,DBLP:conf/icdcs/LunaBBC14,DBLP:conf/icdcn/LunaBBC14,DBLP:conf/opodis/MilaniM15,DBLP:conf/podc/LunaB15}, study the counting problem in the \emph{synchronous} model of dynamic graphs. 
PP can be also represented by dynamic graphs, but is a completely asynchronous model. Moreover, in contrast with the current study, protocols in these works require that all nodes are initialized. These essential differences make the techniques (e.g., termination detection), extensively used there, inappropriate in our case. Note however that their protocols determining the exact count have exponential convergence time, except the results of \cite{DBLP:conf/podc/LunaB15} considering more restricted networks. In some sense this supports the results presented here for weak fairness. \cite{DBLP:conf/podc/LunaB15} presents interesting time complexity lower bounds for counting, but considers networks where anonymous nodes can communicate (by broadcast) \emph{any} amount of data and diffuse it to all other nodes in a \emph{constant} time w.r.t. $\n$, what is of course impossible in our context. 

Due to the difficulty of the problem, a lot of work has been devoted to design protocols counting \emph{approximately} the number of network nodes (see, e.g., \cite{DBLP:journals/pe/GkantsidisMS06,DBLP:journals/jss/KostoulasPGBD07,DBLP:journals/dc/GaneshKMM07,DBLP:conf/imc/RibeiroT10,DBLP:journals/tpds/BaqueroAMJ12,DBLP:conf/cdc/VaragnoloPS10,DBLP:journals/jpdc/AntaMT13}). These protocols use gossiping and probabilistic methods, like probabilistic polling, random walks, epidemic-based approaches, 
 and also exploit classical results on order statistics to infer an estimated number of the nodes. Here, we consider only deterministic protocols for exact counting.

\section{Model and Notations}\label{sec:model}
 A system consists of a collection $\Agents$ of pairwise interacting agents, also called a population. Each agent represents a finite state sensing and communicating mobile device. Among the agents, there is a distinguishable agent called the {\em base station} (BST), which can be as powerful as needed, in contrast with the resource-limited non-BST agents. The 
 non-BST agents are also called \emph{mobile}, interchangeably.
 The size of the population is the number of mobile agents, denoted by $\n$, and is unknown (a priori) to the agents.

A \textit{(population) protocol} can be modeled as a finite transition system whose states are called \textit{configurations}. A \textit{configuration} is 
a vector of states of all the agents.
Each agent has a state taken from a finite set of states, the same for all mobile agents (denoted $Q$), but generally different for BST.  

In this transition system, every transition $C{\rightarrow} C'$ between two configurations $C$ and $C'$ is modeled by a single \textit{transition} between two agents happening during an interaction. That is,
when two agents $x$, in state $p$, and $y$, in state $q$, in $C$, interact (meet),
they execute a transition rule $(p,q) \rightarrow (p', q')$. 
As a result, in $C'$, $x$ changes its state from $p$ to $p'$ and $y$ from $q$ to $q'$. If $p=p'$ and $q=q'$, the corresponding transition is called \textit{null} (such transitions are specified by default), and 
non-null otherwise.\footnote{For simplicity, in some cases,
we do not present protocols under the form of transition rules, but 
under the equivalent form of a pseudo-code.}
\\If there is a sequence of configurations $C = C_0,C_1, \ldots ,C_k = C'$, such that
$C_i \rightarrow C_{i+1}$ for all $i, 0 \le i < k$, we say that $C'$ \textit{is reachable} from
$C$, denoted $C\overset{*}{\rightarrow} C'$.

The transition rules of a protocol are \textit{deterministic}, if for every pair of states $(p, q)$, there is exactly one $(p', q')$ such that $(p, q) \rightarrow (p', q')$. We consider only deterministic transitions and thus, only \emph{deterministic protocols}. Transitions and protocols can be \textit{symmetric} or \textit{asymmetric}. Symmetric means that, if $(p, q) \rightarrow (p', q')$ is a transition rule, then $(q, p)\rightarrow (q', p')$ is also a transition rule. In particular, if $(p, p) \rightarrow (p', q')$ is symmetric, $p'= q'$. Asymmetric is the contrary of symmetric.

Let $(p_1,q_1) \rightarrow (p_2, q_2)$, $(p_2,q_2) \rightarrow (p_3, q_3)$, $\ldots$, $(p_{k-2},q_{k-2}) \rightarrow (p_{k-1},q_{k-1})$, $(p_{k-1},q_{k-1}) \rightarrow (p_k, q_k)$ be the transition rules of a protocol. Then, we shortly write $(p_1,q_1) \overset{*}{\rightarrow} (p_k, q_k)$ to denote a possible sequence of these transition rules, which can be applied (in the same order) on two agents in states $p_1$ and $q_1$, making them interact repeatedly until their states change to $p_k$ and $q_k$, respectively.
We sometimes call an agent in state $p$ a $p$-state agent, or just $p$ agent.
%

An \textit{execution} of a protocol is an infinite sequence of configurations $C_0, C_1, C_2, \ldots$ such that $C_0$ is the starting configuration and 
for each $i\ge 0$, $C_i\rightarrow C_{i+1}$. In a real distributed execution, interactions of distinct agents are independent and could
take place simultaneously (in parallel), but when writing down an execution we can order those
simultaneous interactions arbitrarily.

An execution is said \textit{weakly fair}, if every pair of agents in $\Agents$ interacts infinitely often.
An execution is said \textit{probabilistically fair}, if, for each interaction in the execution, a pair of agents in $\Agents$ is chosen uniformly at random.
An execution is said \textit{globally fair}, if for every two configurations $C$ and $C'$ such that $C\rightarrow C'$, if $C$ occurs infinitely often in the execution, then $C'$ also occurs infinitely often in the execution.
This also implies that, if in an execution there is an infinitely often reachable configuration, then it is infinitely often reached \cite{DBLP:journals/dc/AngluinAER07}.
Global fairness can be viewed as simulating randomized systems without introducing randomization explicitly (any probabilistically fair execution is globally fair with probability 1 \cite{JiangThesis2007}).

A simple example allows to understand better the difference between weak and global (or probabilistic) fairness. Consider a population of 3 agents. Each agent can be white or black, and initially one agent is black and the two others are white. Consider also the protocol in which, when two white agents interact, they become both black and when two agents of different colors interact, they exchange their colors. It is easy to see that there is an infinite weakly fair execution in which there is always one black and two white agents (the black color ``jump'' indefinitely from agent to agent). At the contrary, every globally fair execution terminates in a configuration in which the 3 agents are black, because otherwise there would be infinitely many configurations during an execution from which the ``all black'' configuration could be reached, without ever being reached (contradicting global fairness).

A \textit{problem} is defined by a predicate $\mathcal{D}$ on executions. 
A population protocol $\mathcal{P}$ is said to \textit{solve a problem} $\mathcal{D}$, if and only if every execution of $\mathcal{P}$ satisfies the conditions defining $\mathcal{D}$.
The problem of \textit{counting} is defined by the following condition: eventually, in any execution, there is at least one agent (BST, in our case) obtaining a value of $\n$ in some (estimate) variable (called $c$ in the following) and this value does not change.
Note that the counting predicate has to be satisfied only eventually (and forever after). When it happens,
we say that the protocol has \textit{converged}.
A protocol is called \textit{silent}, if in any execution, eventually, no state of an agent changes \cite{DBLP:journals/acta/DolevGS99}.

In the case of probabilistic fairness, the \textit{convergence time} of a protocol is measured in terms of \textit{parallel time} or \textit{parallel interactions}, i.e., the independent interactions (of distinct agents) occurring in parallel. It is customary to define one \textit{unit of parallel time} as $\n$ consecutive interactions in a probabilistically fair execution. Then, in this case, the convergence time of a protocol is defined by the expected number of parallel time units
in a fair execution till convergence. Moreover, under probabilistic fairness, it appears that technically (in the context of this specific work), we can perform the convergence time analysis in terms of transitions involving BST. It is easy to see that this corresponds to the (asymptotic) expected convergence time in terms of parallel time units.

In the case of weak fairness, the convergence time of a protocol is defined as the maximum number of non-null transitions in a fair execution till convergence.



We consider only \textit{semi-uniform} protocols (cf. \cite{DBLP:journals/dc/DolevIM93,Tel2000}) in the sense that all agents, except BST (whence semi-), are a priori indistinguishable and interact according to the same transition rules. Moreover, 
the protocol functions similarly for any $\n$ and any upper bound $P$ on $\n$.
More formally, we can define a semi-uniform protocol so:
\begin{definition}\label{def:semi-uniform}
A protocol $\mathcal{PP}$ is called \emph{semi-uniform} if for any upper bound $P$ and any execution prefix $e$ in which only agents from a subset $S\subset \Agents$ (including BST) interact, the (standard) projection\footnote{The projection of a configuration on a set of agents $S$ is a restriction of the vector representing the configuration to the elements corresponding to the agents of $S$. Naturally, the projection of an execution $e$ on a set of agents $S$ ($e|_S$) is obtained from $e$ by projecting every configuration of $e$ on $S$, representing an execution where the agents of $\overline{S}$ (the complement of $S$) do not interact.} $e|_S$ of $e$ on the agents of $S$ is an execution prefix of $\mathcal{PP}$ for any bound $P'$ such that $|S|\le P'<P$.
\end{definition}

\begin{remark}
Similarly, if $e$ is an execution prefix of a semi-uniform protocol $\mathcal{PP}$ for $n'\le P'$, it is also an execution prefix of $\mathcal{PP}$ for $n$ s.t. $n'\le n\le P$ and $P'\le P$, if we consider $e$ as an execution prefix for $p$, and then extend the configurations of $e$ with $n-n'$ agents (missing in $e$ and performing no interactions in the extended prefix).
\end{remark}



\section{Time and Space Optimal Counting under Probabilistic Fairness}\label{sec:probabilistic}
\subsection{Time Lower Bound for a Space Optimal Protocol}\label{sec:lower-bound-probabilistic}


Defining time optimality for a counting protocol asks to be cautious. Indeed, a protocol could be efficient for counting some set of agents and slow for counting others. Think of a protocol that ``guesses'' initially a count and checks afterwards whether this count is correct or not. On the right set of agents, this counting protocol would converge in zero time. For other sets it is certainly less efficient than the protocols which estimate the count gradually, starting from 0. We would like to avoid such behavior and thus restrict
%
our attention to protocols having always a ``proof'' that the estimate
they have
corresponds to a 
lower bound on the actual population size (i.e., they have observed a sequence of interactions that justifies this count).
For such protocols, called here ``non-guessing'', the estimate of the size (in the variable $c$) is a non-decreasing function along an execution (and so $c$ is always a lower bound on the number of agents in the population). In the sequel of the section, we consider an arbitrary state-optimal protocol $Count$ and we prove that it converges in expected $\Omega(n\log(n))$ interactions with BST, or equivalently $\Omega(n^2\log(n))$ interactions between agents, or $\Omega(n\log(n)$ parallel time. $Count$ uses (optimally) only two states per mobile agent (with one state, counting is impossible \cite{DBLP:conf/wdag/BeauquierBCS15}). Note also that the result in this section does not assume that $Count$ is symmetric.


We call \emph{trace of an execution prefix} the sequence of transitions of BST in this prefix. Thus, a trace $T$ is a sequence 
of transitions of the form $(s_{BST}, i)\rightarrow(s'_{BST}, j)$, where $s_{BST}$ and $s'_{BST}$ are states of BST, and $i$ and $j$ are states of a mobile agent in the corresponding interaction. Note that this sequence captures all the information that BST has, and completely determines its state.
Let us denote by $x(T)$ the minimal 
population size for which there exists at least one execution prefix with trace $T$.
Thus, since the estimate counter $c$ is non-decreasing, we have $c\leq x(T)$: if it was not the case, by definition of $x(T)$, the protocol could run with $x(T)$ agents and estimate, at some point of the execution, that $c>x(T)$; then, to converge, the protocol would be required to decrease $c$, which is impossible by the assumption above.
We denote by $x_i(T)$ the minimal number of mobile agents in state $i$ (w.l.o.g., $i\in\{0, 1\}$) in a configuration at the end of any possible execution prefix corresponding to trace $T$ (and with $x(T)$ agents).
%
%
Obviously, for all execution prefixes with $n$ agents and a trace $T$, we have $x_0(T)+x_1(T)\leq x(T)\leq n$.

The idea behind Lemma \ref{lem:tbound-rand1} is that, if there is such a transition rule of mobile agents that can decrease the number of agents in state $i$, then when several such agents are present in the population, making them interact and execute this transition will decrease their number to the minimum. The  proof appears in the appendix.

\begin{lemma}\label{lem:tbound-rand1}
If the considered protocol $Count$ has a transition rule that allows to decrease the number of agents in state $i$ through interactions between mobile agents (rules $(i, i)\rightarrow(i, 1-i)$, $(i, i)\rightarrow(1-i, 1-i)$ or $(i, 1-i)\rightarrow(1-i, 1-i)$), then for any trace $T$, we have $x_i(T)\leq 1$.
\end{lemma}

\begin{lemma}\label{lem:tbound-rand2}
Let $T'$ be a trace obtained from a trace $T$ by adding a transition $(s_{BST}, i)\rightarrow(s'_{BST}, j)$ between BST and a mobile agent. If $x(T')>x(T)$, then $x_{1-i}(T)=x(T)$ and $x(T')=x(T)+1$.
\end{lemma}

\begin{proof}
We show the contrapositive: if $x_{1-i}(T)<x(T)$, then $x(T')=x(T)$; if $x_{1-i}(T)=x(T)$, then $x(T')= x(T)+1$.

Let the added interaction in $T'$ be $(s_{BST}, i)\rightarrow(s'_{BST}, i)$:
 \begin{itemize}
\item If $x_{1-i}(T)<x(T)$, some executions with trace $T$ and $x(T)$ agents lead to a configuration with $x(T)-x_{1-i}(T)>0$ agents in state $i$. These agents may interact with BST, so that there still exist executions with trace $T'$ and with $x(T)$ agents, and with the same number of agents in both states. Thus $x(T')=x(T)$, $x_0(T')=x_0(T)$, and $x_1(T')=x_1(T)$;
\item If $x_{1-i}(T)=x(T)$ (which implies that $x_i(T)=0$), all executions with trace $T$ and $x(T)$ agents contain no agent in state $i$. Thus, $x(T')=x(T)+1$ (it cannot be higher, since one can build an execution with $x(T)$ agents interacting in pattern resulting in trace $T$, and an extra agent in state $i$ that does not interact, then released for this last interaction). Since, as a result of this interaction, an agent in state $i$ remains, if no rule $(i, 1-i)\rightarrow(1-i, 1-i)$ can remove it, $x_i(T)=1$, else $x_i(T)=0$. In addition, $x_{1-i}(T')=x_{1-i}(T)$, because the execution described above with $x(T')=x(T)+1$ agents results in a configuration with $x_{1-i}(T)$ agents in state $i$. Finally, we have: $x(T')=x(T)+1$, $x_i(T')\leq 1$, and $x_{1-i}(T')=x_{1-i}(T)$.
\end{itemize}

Now, 
let the added interaction in $T'$ be $(s_{BST}, i)\rightarrow(s'_{BST}, 1-i)$:
\begin{itemize}
\item If $x_{1-i}(T)<x(T)$, some executions with trace $T$ and $x(T)$ agents contain agents in state $i$. These agents may interact with BST, so that there still exist executions with trace $T'$ and with $x(T)$ agents, and with an agent in state $i$ that has changed its state. Thus $x(T')=x(T)$, $x_i(T')=\max\{x_i(T)-1, 0\}$, and $x_{1-i}(T')\leq x_{1-i}(T)+1$.
\item If $x_{1-i}(T)=x(T)$ (which implies that $x_i(T)=0$), all executions with trace $T$ and $x(T)$ agents contain no agent in state $i$. Thus, trace $T'$ cannot be achieved with $x(T)$ agents, and $x(T')\geq x(T)+1$. $T'$ can be achieved by adding an extra agent in state $i$ during trace $T$ and releasing it to realize the last interaction, so that $x(T')=x(T)+1$. An agent in state $i$ has its state changed to $1-i$, so that $x_i(T')=\max\{x_i(T)-1, 0\}$, and $x_{1-i}(T')\leq x_{1-i}(T)+1$.
\end{itemize}

Thus, $x$ can increase only as the result of an interaction of BST with an agent in state $i$ such that $x_{1-i}=x$. After this interaction, to increment $x$ again, BST must increment $x_{1-i}$, which it can do only by switching an agent state to $1-i$.
\end{proof}

In particular, this implies that, if interactions between mobile agents can decrease the number of agents in both states 0 and 1, $x_i\leq 1$ for $i\in\{0, 1\}$, and $x\leq 2$, i.e., any trace can be obtained with two agents only, and $Count$ is incorrect.

\begin{theorem}
A two-state non-guessing counting protocol $Count$ (correct under probabilistic fairness) converges in  $\Omega(n\log n)$ 
expected interactions with BST (equivalently, in $\Omega(n\log n)$ expected parallel time).
\end{theorem}
\begin{proof}
Consider a converging execution of this protocol with $n\geq3$ mobile agents. Denote by	$T$ the trace of this execution until $x(T)=n$ (which occurs, since the protocol converges), and by $T'$ the trace obtained from $T$ by removing its last interaction. Denote by $i$ the mobile agents state such that $x_i(T')=x(T')=n-1>1$ (this state exists by Lemma~\ref{lem:tbound-rand2}).

Thus, $x_i$ must increase from 0 to $n$. Now, $x_i$ increases only when BST meets an agent in state $1-i$ and changes its state to $i$, and the number of mobile agents in state $1-i$ can only increase in an interaction with BST (from Lemma~\ref{lem:tbound-rand1}, as $x_i(T')>1$, no interaction between mobile agents can increase the number of agents in state $1-i$). Thus, in the last configuration before convergence with $x_i=0$, all agents are in state $1-i$ (because $x_i$ eventually reaches $n$, and can increase only when a mobile agent state is switched from $1-i$ to $i$).

Thus, in any converging execution, there is a configuration in which all agents are in state $1-i$, and then all agents are switched to state $i$ in interactions with BST (except possibly one, that has been counted by BST, and will be switched to state $i$ by it only in a further interaction). Now, if $k$ agents are in state $i$, the probability for BST to meet an agent in state $1-i$ in its next interaction is $\frac {n-k}n$, and the expected number of interactions (involving BST) before meeting an agent in state $i$ and incrementing $x$ is $\frac n{n-k}$.

So, the expected length of an execution before convergence is $\geq\sum_{k=0}^{n-1}\frac{n}{n-k}=\sum_{l=1}^{n}\frac{n}{l}=n\sum_{l=1}^{n}\frac{1}{l}=nH_n$, with $H_n$ the $n$th harmonic number. It is known that $H_n=\Theta(\log n)$, hence the result.

Given that there are $n$ agents in the system, one interaction out of $n$ involves BST in average. Hence, the expected $\Omega(n\log n)$ interactions with BST are equivalent to expected $\Omega(n^2\log n)$ interactions between agents, and to the expected $\Omega(n\log n)$ parallel time.
\end{proof}

\subsection{Time and Space Optimal Protocol (Prot. \ref{alg:counting2s-time-opt})}\label{sec:time-space-opt-prot}

The one bit space optimal protocol of \cite{DBLP:conf/wdag/BeauquierBCS15} is recalled in the appendix (Prot. \ref{alg:counting2s}) together with its time complexity analysis that gives the average convergence time of $\Theta(2^\n)$ interactions. In this section, we modify Prot. \ref{alg:counting2s} to obtain an (asymptotically) time optimal protocol, Prot. \ref{alg:counting2s-time-opt}, converging in $O(n\log n)$ time, and still optimally using only one bit of memory per mobile agent. We present and prove this protocol and its convergence time below.

In Prot. \ref{alg:counting2s-time-opt}, each mobile agent can be in one of two states $0$ or $1$, and respectively called $0$ or $1$ agent. We write $c_0$ and $c_1$ for the protocol's count of $0$ and $1$ agents resp., and $n_0$ and $n_1$ for the actual number of $0$ and $1$
agents in the population.  
The total number of agents is then $n = n_0 +
n_1$ 
and the base station's
estimate of $n$ is $c = c_0 + c_1$.  The values $c_0$ and $c_1$ are
both initialized to $0$. They may be seen as the implementations of the $x_0$ and $x_1$ used in the lower bound proof in the previous subsection.

The modified protocol proceeds in alternating phases.  In a
\concept{zero phase}, BST only converts zeros to ones.
Whenever it does so, it decrements $c_0$ if it is positive and
increments $c_1$.
In a \concept{one phase}, it does the reverse.  We start the protocol
in a zero phase.

The same argument as
for the original protocol shows that $c_b \leq n_b$ holds as an
invariant and that $c = c_0 + c_1$ is non-decreasing over time (Lemma 1 in \cite{DBLP:conf/wdag/BeauquierBCS15}).  If, in
addition, we can stay in two phases long enough that every agent is
converted from $b$ to $1-b$ in the first phase, and then every agent is
converted from $1-b$ to $b$ in the second phase, at the end of the
second phase we will have $c = n$, giving convergence.

Let us now specify when BST switches
between phases.  Suppose that BST is going to start in a $b$ phase.
We adopt the following procedure (in two stages):
\begin{enumerate}
\item \emph{(pre-phase)} Flip any $b$ agents we encounter to $1-b$ as long as $c_b > 0$.
\item \emph{(the phase itself) }Continuing flipping any $b$ agents BST encounters to $1-b$ until it
sees $6 (c_b \ln c_b + 1)$ agents marked $1-b$ in a row without seeing an
agent marked $b$.  If this event occurs, flip the phase (switch to the $1-b$ phase).
\end{enumerate}

The first rule (the pre-phase) guarantees that whenever we start a $b$ phase,
$c_{1-b}$ is always zero.  This in turn guarantees that $c_b$ is never
lower at the start of a $b$ phase than it is at the start of any
previous $b$ phase.

\begin{algorithm}[h]
\caption{-- Time and Space Optimal Counting under Probabilistic Fairness}
\begin{algorithmic}
    \STATE
	\STATE \textbf{Variables at BST:}
	\INDSTATE $\size_0$: non-negative integer, initialized to $0$; eventually holds $n_0$
    \INDSTATE $\size_1$: non-negative integer, initialized to $0$; eventually holds $n_1$
	\INDSTATE $\sizeT$: non-negative integer initialized to $0$; eventually holds $n$
    \INDSTATE $\cnt$: non-negative integer initialized to $0$
    \INDSTATE $\phase$  $\in \{0,1\}$, initialized to $0$
	\STATE \textbf{Variable at a mobile agent $x$:}
	\INDSTATE $b$ $\in \{0,1\}$, initialized \emph{arbitrarily}
	\STATE
\end{algorithmic}

\begin{algorithmic}[1]

\WHEN{\underline{a mobile agent $x$ with mark $b$ interacts with BST}}
\IF{$b = \phase$}
    \STATE{$\cnt \leftarrow 0$}
	\IF{$\size_{b}>0$}
		\STATE $\size_{b} \leftarrow \size_{b}-1$\label{l-2s:size-mark--}
	\ENDIF
	\STATE $b \leftarrow 1 - b$ \label{l-2s:change-mark}
    \STATE $\size_{b} \leftarrow \size_{b}+1$ \label{l-2s:size-mark++}
\ELSIF{$\cnt \ge 6(\size_{b}\ln\size_{b}+1)$}
    \STATE{$\cnt \leftarrow 0$}
    \STATE{$\phase \leftarrow 1-\phase$}
\ELSIF{$\size_{\phase}=0$} 
    \STATE{$\cnt \leftarrow \cnt +1$}
\ENDIF
    \STATE $\sizeT \leftarrow \size_0+\size_1$ \label{l-2s:size-update}
\ENDWHEN
\end{algorithmic}\label{alg:counting2s-time-opt}
\end{algorithm}

For the convergence bound, begin by bounding the likely length of a phase:
\begin{lemma}
\label{lemma-phase-length}
Each phase requires $O(n \log n)$ parallel interactions with high probability.
\end{lemma}
\begin{proof}
Suppose we are in a $b$ phase.
Using standard bounds on the Coupon Collector Problem 
it holds with high probability that BST has interacted
with every agent after $O(n \log n)$ interactions.  So either the
phase has already ended, or every agent now carries $1-b$.  In the
latter case, the phase can run for at most $6 (n \ln n + 1) = O(n \log
n)$ interactions before ending.
\end{proof}

We now show that, on average, the protocol executes $O(1)$ phases.
This requires the following technical lemma showing that BST finds all $b$ agents in a $b$ phase if there are enough to
begin with.
\begin{lemma}
\label{lemma-all-flip}
If phase $b$ starts with $n_b \geq n/2$, then it ends with $n_b = 0$
with probability at least $1/2$.
\end{lemma}
\begin{proof}
For simplicity we will assume $b=0$; the $b=1$ case is symmetric.  So
we are looking at a zero phase that starts with $n_0 \geq n/2$.  From the
structure of the protocol, we know that at the start of this phase,
$c_1 = 0$, but $c_0$ might be larger.  It happens that the worst
case is when $c_0 = 0$, but we will analyze the process for any
initial value of $c_0$.

In the analysis below we will fix $n_0, n_1$,
to their values at the start of the phase.  To keep track of
what happens, let $i$ be the number of zero values converted to ones
so far during this phase; given the value of $i$, this gives $n_0-i$
zeros and $n_1+i$ ones in the population, and the value of the $c_1$
register will be $i$.  We fail to convert all zeros to ones if we exit
the phase while $i$ is less than $n_0$.

For each particular value of $i$, this occurs only if (a) $c_0$ is
already $0$, and (b) BST observes $6 (n \ln n + 1)$ ones
in a row.  Whether or not $c_0 = 0$, the latter event occurs with probability exactly
\begin{equation}
\left(\frac{n_1+i}{n}\right)^{6 (i \ln i + 1)}
\end{equation}
which by the union bound gives an upper bound on the probability
that we leave the phase for any $i < n_0$ of
\begin{equation}
\sum_{i=0}^{n_0-1} \left(\frac{n_1+i}{n}\right)^{6 (i \ln i + 1)}
\end{equation}

We will bound this sum by considering the terms with $i < n_0/2$ and
$i \geq n_0/2$ separately. The detailed computations for each case appear in the appendix and give a bound of $2/5$ for the case $i < n_0/2$, and $1/200$ for $i \geq n_0/2$.
%
%
%
The original sum is thus bounded by $2/5+1/200 < 1/2$ for all
$n>0$, giving the claimed bound.
\end{proof}

\begin{theorem}
    \label{theorem-convergence}
    The modified protocol (Prot. \ref{alg:counting2s-time-opt}) converges to $c=n$ in an expected $O(n \log n)$ interactions with BST (equivalently, in an expected $O(n \log n)$ parallel interactions).
\end{theorem}
\begin{proof}
There are two cases, depending on the initial value of $n_0$.
\begin{enumerate}
    \item If $n_0 \geq n/2$ in the starting configuration, then $n_0 \geq n/2$ at the
        start of each zero phase.  From Lemma~\ref{lemma-all-flip},
        BST converts all zeros to ones in any of these
        phases with probability at least $1/2$.  If this event occurs,
        the following one phase converts all ones to zeros with
        probability at least $1/2$ as well, giving a probability of at
        least $1/4$ for each pair of phases that we converge to the
        correct count.  Thus the protocol converges in an expected
        $4\cdot 2 = 8$ phases.
    \item If $n_0 < n/2$, then the initial zero phase ends with at
        least $n/2$ ones (because any conversion during this phase
        can only increase the number of ones).  So the first one phase
        starts with $n_1 = n/2$.  Repeating the above analysis shows
        that the protocol converges after at most $8$ phases on
        average on top of the initial zero phase, giving an expected
        $9$ phases total.
\end{enumerate}
Because each of these $O(1)$ phases takes $O(n \log n)$ expected interactions
(Lemma~\ref{lemma-phase-length}), this gives a total expected number
of interactions of $O(n \log n)$.
\end{proof}

\section{Time Lower Bound for Space Optimal Counting under Weak Fairness}\label{sec:lower-bound-weak}

To obtain this lower bound we first prove properties that have to be satisfied by any space optimal symmetric counting protocol functioning under weak fairness. These properties are important by themselves, as they can be useful in future studies of counting under weak fairness in PP. For instance, 
Proposition \ref{prop-count-weak-property-naming}, states that a 
counting protocol has to distinctly name all the agents in any population of size $\n< P$. Recall that $P$ is the upper bound on the size $\n$ of the population.

Next, from Proposition \ref{prop-count-weak-property-naming} and Lemma \ref{lem-count-weak-property-sink1}, it easily follows that any symmetric counting protocol under weak fairness has to use at least $P$ states per mobile agent (to be able to count any population of at most $P$ agents). This gives a somewhat simpler proof than the original one in \cite{Beauquier2007}. 

The next important property, given in Proposition \ref{prop-count-weak-property-sink2}, is that any space optimal symmetric counting protocol under weak fairness has a unique ``sink'' state $m$ s.t., for every possible state $s\in Q$ of a mobile agent, there is a transition sequence $(s, s)\overset{*}{\rightarrow}(m,m)$, with $(m, m){\rightarrow}(m,m)$ and $m$ cannot be one of the distinct names given by the protocol in case $\n< P$.

The results above show in particular that, for any considered space-optimal counting protocol, if mobile agents are not named yet, agents in state $m$ will continue to appear (for any $P$ and $\n\le P$).

Moreover, recall that we consider counting with non-initialized mobile agents. In this case, to overcome the known impossibility \cite{Beauquier2007}, we assume one initialized and distinguishable agent BST that eventually counts the other $\n$ (mobile) agents. Note that having a distinguishable agent is necessary. To see this, consider a starting configuration where all agents start at the same state (which has to be equal to the state of the only initialized agent). If the size of the population is even, by Proposition \ref{prop-count-weak-property-sink2}, there is a weakly fair execution reaching and staying in the configuration where all agents are in the ``sink'' state $m$, and by Proposition \ref{prop-count-weak-property-naming}, no counting can be realized.

Using the above-mentioned properties and those of semi-uniform protocols (see Definition \ref{def:semi-uniform} in Sec. \ref{sec:model}), we prove the lower bound given in Theorem \ref{th-count-weak-tcomplexity}. This is one of the main results of the paper. It shows that, under weak fairness, counting undistinguishable, state-optimal and non-initialized agents in symmetric PP is a costly task in terms of a convergence time. It takes $\Omega(2^{\n})$ non-null transitions.  

Finally, we show (Proposition \ref{prop-2P:time} in the appendix)  that the space optimal protocol for weak fairness presented in \cite{DBLP:conf/wdag/BeauquierBCS15} converges in $\Theta(2^\n)$ non-null transitions. This proves that this is a time optimal protocol among all the space optimal semi-uniform protocols, under weak fairness.

\begin{proposition}\label{prop-count-weak-property-naming} Let $Count$ be a 
(silent or not) counting protocol correct under weak fairness (for any $\n\le P$). For any weakly fair execution $e = C_1, C_2, C_3, \ldots, C_j, \ldots$ of $Count$ on a population $\Agents$ of size $n< P$, there is an integer $k$ such that, for any $j\ge k$, no two mobile agents are in the same state in $C_j$.\end{proposition}

\begin{proof} Let us assume, by contradiction, that in $e$, there are infinitely many configurations with two agents in the same state. Since the state space is finite and the number of agents too, two specific agents $x_2$ and $x_3$ from $\Agents$ are necessarily simultaneously in some state $s\in Q$ in infinitely many configurations. Let $C_{j_1}, C_{j_2}, C_{j_3},\ldots$ be these configurations such that $e = e_1, C_{j_1}, e_2, C_{j_2}, e_3, C_{j_3},\ldots$ W.l.o.g., we choose these configurations such that, in every execution segment $e_i$, every agent in $\Agents$ interacts with every other (this is possible with weak fairness).

Now consider a population $\Agents'=\Agents\cup \{x_1\}$ of size $n+1$. To prove the proposition, we will construct a weakly fair execution $e'$ of $Count$ in population $\Agents'$ where no agent can distinguish $e'$ from $e$, and where consequently $Count$ wrongly counts only $n$ agents instead of the existing $n+1$.

We construct $e'$ based on $e$. 
First, we assume that in $e'$, $x_1$ is in state $s$ in the starting configuration, and  $e'=e'_1, C'_{j_1}, e'_2, C'_{j_2}, e'_3, C'_{j_3},\ldots$ such that each segment $e'_i$ follows the same transition sequence as in $e_i$, but where the agents $x_2$ or $x_3$ participating in the corresponding interactions can be replaced by $x_1$ in the appropriate state, as we explain below. We ensure also that in every $C'_{j_i}$, each of the three agents $x_1, x_2, x_3$ is in state $s$.

More precisely, in segment $e'_{3r+1},C'_{j_{3r+1}}$ ($r \ge 0$), agent $x_1$ does not interact with the rest of the agents, and all the others interact exactly as in $e_{3r+1},C_{j_{3r+1}}$ (each agent $x_i$ is in state $s$ in $C'_{j_{3r+1}}$).  In $e'_{3r+2},C'_{j_{3r+2}}$, agent $x_2$ does not interact with the rest of the agents, and $x_1$ replaces $x_2$ in all the interactions where $x_2$ interacts in $e_{3r+2},C_{j_{3r+2}}$, and all the others interact as in $e_{3r+2}$ (but with $x_1$ instead of $x_2$ in the corresponding interactions).  Each agent $x_i$ is in state $s$ in $C'_{j_{3r+2}}$. In $e'_{3r+3},C'_{j_{3r+3}}$, agent $x_3$ does not interact with the rest of the agents, and $x_1$ replaces $x_3$ in all the interactions where $x_3$ interacts in $e_{3r+3},C_{j_{3r+3}}$, and all the others interact as in $e_{3r+3}$ (but with $x_1$ instead of $x_3$ in the corresponding interactions).  Each agent $x_i$ is in state $s$ in $C'_{j_{3r+3}}$.

We emphasize again that $e'$ is possible, because in every $C'_{j_i}$, each of the three agents $x_1, x_2, x_3$ is in state $s$, so any of them can replace any other in the transition sequence of the next segment $e_{i+1}$. Moreover, $e'$ is weakly fair, because each agent $x_i$ interacts with all the other agents in the appropriate $e'_i$ segments (and by the assumption on $e_i$), and other agents too, due to the weak fairness of $e$. Finally, in $e'$, every agent from $\Agents$ (including BST), executes exactly the same sequence of transition rules as it does in $e$, so no agent can distinguish the fact that the population is actually $\Agents'$ with $n+1$ agents, and $Count$ counts only $n$ agents as it does in $e$. This is a contradiction to the assumption that $Count$ is a correct counting protocol.
\end{proof}

The proof of Lemma \ref{lem-count-weak-property-sink1} uses similar techniques as the proof of Proposition \ref{prop-count-weak-property-naming} and appears in the appendix.

\begin{lemma}\label{lem-count-weak-property-sink1} Let $Count$ be a symmetric (silent or not) counting protocol correct under weak fairness (for any $\n\le P$). Consider any weakly fair execution $e = C_1, C_2, C_3, \ldots, C_j, \ldots$ of $Count$ on a population $\Agents$ of size $n< P$. There is an integer $k$ such that, for any $j\ge k$, no mobile agent is in a state $m\in Q$ such that there is a possible sequence of transitions of $Count$ $(m,m)\overset{*}{\rightarrow}(m,m)$.\end{lemma}

The following two propositions follow from Proposition \ref{prop-count-weak-property-naming} and Lemma \ref{lem-count-weak-property-sink1}. A proof of Proposition \ref{prop-count-weak-Ps} uses similar techniques as the proof of Proposition \ref{prop-count-weak-property-sink2} and appears in the appendix.

\begin{proposition}\label{prop-count-weak-Ps}
Any symmetric counting protocol $Count$ correct for any $\n\le P$ (undistinguishable and non-initialized) mobile agents, under weak fairness, has to use at least $P$ states per mobile agent.
\end{proposition}


\begin{proposition}\label{prop-count-weak-property-sink2}
Consider any symmetric (silent or not) counting protocol $Count$ correct under weak fairness (for any $\n\le P$), and using at most $P$ states per mobile agent. For every state $s\in Q$, there is a transition sequence $(s, s)\overset{*}{\rightarrow}(m,m)$, s.t. $m$ is unique and does not appear infinitely often in executions with $\n< P$. Moreover, $(m, m){\rightarrow}(m,m)$.
\end{proposition}

\begin{proof}
As $Count$ is symmetric, any two agents, both in some state $s\in Q$, in an interaction, have to execute a symmetric transition of the form $(s,s)\rightarrow(s_1,s_1)$. Thus there is a possible sequence of transitions $(s,s)\rightarrow(s_1,s_1)\rightarrow(s_2,s_2)\rightarrow(s_3,s_3)\ldots$ As mobile agents are finite state, for some $j> i\ge 1$, $s_{i}=s_{j}$, i.e. $(s_i, s_i)\overset{*}{\rightarrow}(s_i,s_i)$.  By Lem. \ref{lem-count-weak-property-sink1}, $s_i = m$ s.t. $m$ does not appear infinitely often in executions with $\n< P$. As there are at least $P-1$ states appearing infinitely often in an execution with $\n= P-1$ (by Proposition \ref{prop-count-weak-property-naming}), there is at most one such possible state $m$ in a $P$ state protocol. Thus, the first part of the lemma holds.

Finally, by contradiction, if $(m, m){\rightarrow}(s,s)$ s.t. $s\neq m$, then the previous part of the proof 
implies $(s, s)\overset{*}{\rightarrow}(s,s)$. 
When $\n= P-1$, and as $Count$ uses only $P$ states, and $m$ never appears infinitely often in an execution, $s$ does appear infinitely often in configurations of an execution (by Proposition \ref{prop-count-weak-property-naming}). This is a contradiction to Lem. \ref{lem-count-weak-property-sink1}. Thus, $(m, m){\rightarrow}(m,m)$.
\end{proof}

To prove the lower bound (Theorem \ref{th-count-weak-tcomplexity}), in addition to the results above, we use the following definitions related to the considered counting protocols. 

\begin{definition}\label{def:weak-fairness-bound}~
\begin{itemize}
\item We call \emph{homonyms}, or \emph{homonymous agents}, mobile agents in the population having the same state, but different from $m$.
\item We say that two (or more) homonyms (in state $s$) are \emph{reduced (to $m$)} whenever a sequence of transitions $(s, s)\overset{*}{\rightarrow}(m,m)$ is applied to them.
\item We say that a mobile agent is \emph{named} if it has a state different from $m$ (a \emph{name}). A group of agents is \emph{named} if each of them is named with a distinct name. Similarly, a configuration of agents is named, if all the agents in this configuration are named.
\item A \emph{reduced (from homonyms) configuration} is a configuration without any homonym. 
    Given a configuration $C$, let $R(C)$ be the set of names of mobile agents appearing an odd number of times in $C$, i.e., the set of names appearing in the corresponding reduced configuration.
\item For any 
two sets $E, E'\subseteq \{1,\ldots,n\}$, we denote by $E\triangle E'\equiv E\cup E'- E\cap E'$ their symmetric difference. In particular, $E\triangle\{e\}$ ($e\in\{1,\ldots,n\}$) is $E\cup\{e\}$ if $e\notin E$, and $E-\{e\}$ if $e\in E$.

\item A \emph{stationary point (or state) of BST} is a state $s_{BST}$ of BST such that $(s_{BST}, m)\rightarrow(s'_{BST}, m)$ and $s'_{BST}$ is also stationary. (Note that this transition sequence can be broken, i.e., BST can change its state to a non-stationary one, after an interaction with an agent in a state $s\neq m$.)
\item The \emph{naming sequence} 
 is a sequence $U=(s_{BST}^i, s_i)_{i\geq 1}$ of pairs $(s_{BST}^i, s_i)$, where $s_{BST}^i$ is a BST state and $s_i$ is a mobile agent state. $U$ is defined inductively as follows: $s_{BST}^1$ is the initial state of BST, and for any $i\geq 1$, $(s_{BST}^{i}, m)\rightarrow (s_{BST}^{i+1}, s_i)$ is a transition rule of the protocol.
Let $U_j$ be a prefix of $U$ s.t. $U_j = (s_{BST}^i, s_i)_{1\leq i\leq j}$.
\end{itemize}
\end{definition}

To obtain the result of Theorem \ref{th-count-weak-tcomplexity}, we focus on the set of the longest execution prefixes where BST meets and names agents in state $m$ (according to the fixed naming sequence, defined in Def. \ref{def:weak-fairness-bound}). In such prefixes, we study the possibility of the occurrence of a stationary point (Def. \ref{def:weak-fairness-bound}). We show that, for a semi-uniform counting protocol (Definition \ref{def:semi-uniform}, Sec. \ref{sec:model}), for any $n<P$, such a point does not exist (Lemma \ref{lem:count-weak-nonstat}), before BST has named all the $n$ agents. That is, in the case of the considered executions, BST will continue giving names to $m$-state agents that it meets (it won't be ``blocked'' waiting for some named agent, for possibly deciding to change its strategy accordingly).

So, using Lemma \ref{lem:count-weak-nonstat}, we prove that there exists an execution prefix in which BST continuously meets and names agents in state $m$ without entering a stationary state. Then, 
we show that the longest such prefix corresponds to the naming sequence of length $\Omega(2^{n})$. This is by observing that the number of starting unnamed (reduced) configurations is $2^{n}-1$, for $n=P-1$, and that after each step (transition) of the naming sequence, BST can accomplish naming of only one such starting configuration. 
To prove the result for any $n$ and $P$, we use Definition \ref{def:semi-uniform} of semi-uniform protocols. Intuitively, for such a protocol, when only a subset of a population of $x<n$ agents interacts with BST, BST should behave like $P=x$ is possible, even if it is larger. 

%
%
In the following lemma, we show that the naming sequence does not contain any stationary state.


\begin{lemma}\label{lem:count-weak-nonstat}
Let $Count$ be a symmetric (silent or not) semi-uniform
counting protocol correct under weak fairness (for any $\n\le P$) and using $P$ states per mobile agent. 
The naming sequence $U$ (Def. \ref{def:weak-fairness-bound}) of $Count$ 
does not contain any stationary state.
\end{lemma}
\begin{proof}

Assume by contradiction that the first stationary state in $U=(s_{BST}^i, s_i)_{i\geq 1}$ is in the $k^{th}$ step (element), i.e., the state $s_{BST}^k$ is stationary, for some $k>0$. 
In the following, we assume that $P> 2k$ (we justify this assumption later) and for any $j\geq0$, we construct an execution $e_j$ of $Count$ for a population of $2k+j$ $(< P)$ agents. In $e_j$, $k+j$ agents are initially in state $m$, while the $k$ other agents are in states $s_1$, $s_2$, \ldots, $s_k$ (these are the states $s_i$ appearing in $U$). 
Each $e_j$ is composed of two phases.

In the first phase, let 
$k$ $m$-state agents interact one by one with BST, and at the end of this phase (by the definition of $U$) the states of these agents are $s_1$, $s_2$, \ldots, $s_k$. Notice that now among these $k$ agents (call it the first sub-population), those that have names have homonyms in the second sub-population (of other $k+j$ agents). Notice also that, by Definition \ref{def:semi-uniform} of semi-uniform protocols, this is also a prefix of an execution projected on a population of only $k$ agents. By the same property, this can be also a prefix of execution for a larger population and for an upper bound $P$ as large as we want. Thus, $P>2k$ is a valid assumption.

In the second phase, let every named (let us say, by $s$) agent of the first sub-population interact with its homonym in the second sub-population, s.t. the sequence of transitions $(s, s)\overset{*}{\rightarrow}(m, m)$ takes place for each such pair of homonyms (possible by Proposition \ref{prop-count-weak-property-sink2}). At the end of this second phase, all agents are thus in state $m$ (at most $2k$ were homonyms at the end of the first phase, and $j$ were initially in state $m$ and never interacted). 

Now, no matter how the agents continue to interact, as all mobile agents are in state $m$, and BST is in a stationary state $s_{BS}^{k}$, all mobile agents remain forever in state $m$  (Proposition \ref{prop-count-weak-property-sink2}) and BST remains in a stationary state (by definition of such a state). Thus, after the second phase, we make all pairs of agents interact infinitely often to obtain a weakly fair execution for any $j$. 
In all these constructed executions, no agent can distinguish between the executions and the corresponding population sizes ($2k+j< P$ agents for any $j\geq0$), and thus a correct counting cannot be obtained.

By Definition \ref{def:semi-uniform} of a semi-uniform protocol, 
the projection of the first phase (of any $e_j$) on the first sub-population (a group of $k$ agents interacting with BST in this phase), is also a prefix of an execution of $Count$ for a bound $P'\ge k$ ($k\le P'\le P$). Thus, for any such $k$ and $P'$, there is no stationary point in $U_k$.
In other words, the lemma holds for any value of $k$ and $P$, and thus also for any prefix of $U$. 
\end{proof}

\begin{theorem}\label{th-count-weak-tcomplexity} Let $Count$ be a symmetric (silent or not) semi-uniform counting protocol correct under weak fairness (for any $\n\le P$) and using $P$ states per mobile agent.
The convergence time of $Count$ is at least $2^{\n} -1$ non-null transitions.
\end{theorem}

\begin{proof}
We will build an execution of $Count$ where the length of the corresponding naming sequence $U$ before convergence (and thus the convergence time of $Count$)
is at least $2^{n}-1$.

Consider a population of $n=P-1$ agents. 
Consider a possible execution prefix $e$ where BST interacts only with $m$-state agents as long as they are not distinctly named.
If the agents are not distinctly named, by Proposition \ref{prop-count-weak-property-sink2}, there is always either at least one agent in state $m$, or some homonyms that can be reduced to $m$. So, assume that in $e$, whenever the mobile agents are not distinctly named and there is no agent in state $m$, a reduction of some homonyms is done. Then, an agent in state $m$ interacts with BST. By Lemma \ref{lem:count-weak-nonstat}, in every such corresponding transition, the state of BST is not stationary, and thus eventually an $m$-state mobile agent is ``given'' a name by BST. Let us repeat this scenario, until all the $n$ agents are named. This is an execution prefix $e$ (with $n=P-1$) that we consider below.

By Lemma \ref{lem:count-weak-nonstat},  $Count$ has to name $n=P-1$ mobile agents, starting from any unnamed configuration $C_j$.
For this specific case of $n=P-1$, there is \emph{exactly one} possible configuration $C^*$ (ignoring the state of BST and the permuted configurations\footnote{A permuted configuration obtained by permuting the elements in the configuration vector.}) where all $n$ mobile agents are named. Notice that $R(C^*)=Q\setminus m$ ($R()$ is defined in Def. \ref{def:weak-fairness-bound}). 
Consider a prefix $U_j= (s_{BS}^i, s_i)_{0\leq i\leq j}$ of the unique naming sequence $U$. Notice also that, by Def. \ref{def:weak-fairness-bound}, $A\triangle B= A'\triangle B$ iff $A= A'$.  This implies that, given $U_j$, there is a unique  $R(C_j)$ such that $R(C_j) \triangle\{s_1\}\triangle \{s_2\}\triangle \{s_3\} \ldots\triangle \{s_j\} = \{a_1, a_2, \ldots, a_{P-1}\}=R(C^*)=Q\setminus m$. 
As $|\{R(C_j): C_j \text{ is any possible unnamed starting configuration}\}|=2^{P-1} -1$, 
the length of $U$ is at least $2^{P-1} -1=2^{n} -1$. Hence, the length of the longest execution $e$ is at least $2^{n-1} -1$ s.t. $n=P-1$.  
By the Remark following Definition \ref{def:semi-uniform}, $e$ is also an execution prefix for any bound $\ge P$, given the same population size. Thus, the theorem actually holds for any $n$ and any upper bound $P$ on $n$.
%
\end{proof}

\section{Conclusion and Perspectives}

This work is a sequel to \cite{DBLP:conf/wdag/BeauquierBCS15} and it answers the questions concerning time complexity of the symmetric space optimal protocols proposed there. 
What can be learned from the current work is that there exists a big difference, not only in terms of the required space, but also in terms of time complexity, between the case where the interactions between agents are random and the case where they are only weakly fair.  
%

From a more practical point of view, it would be interesting to investigate if 
this difference still exists 
when more memory space is given to the agents.
We already know that, concerning weak fairness, a supplementary bit ($2P$ states) allows to design protocols like in \cite{DBLP:journals/tcs/IzumiKIW14} with a logarithmic round complexity (a round being a shortest fragment of execution where each agent interacts with each other), while another additional bit allows to solve this problem in only constant number of rounds \cite{Beauquier2007}.
Concerning global or probabilistic fairness, there exist less studies about counting protocols and especially about their complexity analysis. For example, it would be certainly interesting to determine which size of memory is needed, for having an expected constant convergence time.
More generally, studying formally the trade-offs between space and time complexities for counting algorithms in population protocols could be a valuable sequel to the present work.

\bibliographystyle{is-plain}
\bibliography{biblio_janna}


\appendix
\section*{Appendix}

\noindent{\textcolor{darkgray}{$\blacktriangleright$}\nobreakspace\sffamily\bfseries  Lemma \ref{lem:tbound-rand1}.}
\emph{If the considered protocol $Count$ has a transition rule that allows to decrease the number of agents in state $i$ through interactions between mobile agents (rules $(i, i)\rightarrow(i, 1-i)$, $(i, i)\rightarrow(1-i, 1-i)$ or $(i, 1-i)\rightarrow(1-i, 1-i)$), then for any trace $T$, we have $x_i(T)\leq 1$.}

\begin{proof}
Consider such a protocol, with a rule $(i, i)\rightarrow(i, 1-i)$ or $(1-i, 1-i)$. Consider a trace $T$. This trace can be obtained by an execution prefix $e=C_1, C_2, \ldots, C_k$, such that in configuration $C_k$, there are $\ell$ agents in state $i$, and $x(T)-\ell$ agents in state $1-i$. Now, we can expand this execution prefix by making agents in state $i$ interact, until they all are in state $1-i$, except one in the case of the interaction $(i, i)\rightarrow (i, 1-i)$. This new execution prefix contains the same interactions with BST, and thus, has the same trace. The last configuration of this execution prefix contains at most one agent in state $i$. So, $x_i(T)\leq 1$.


Now, if $(i, 1-i)\rightarrow (1-i, 1-i)$ is the only rule allowing the number of agents in state $i$ to decrease, to use the same kind of reasoning, one must first show that there can always be an agent in state $1-i$ in the configuration, to allow this rule to be executed. So let us show that a situation with $x_i(T)=x(T)>1$ cannot be reached. Indeed, consider a trace $T$ such that $x(T)=x_i(T)=1$. At least one execution with this trace can lead to a configuration with a single agent in state $i$. In this configuration, transition $(s_{BST}, i)\rightarrow(s'_{BST}, i)$ is always possible, so that either transition $(s_{BST}, 1-i)\rightarrow(s'_{BST}, *)$ or $(s_{BST}, i)\rightarrow(s'_{BST}, 1-i)$ must eventually occur for $x(T)$ to increase. When the transition $(s_{BST}, 1-i)\rightarrow(s'_{BST}, *)$ takes place, all executions with this trace contain at least one agent in state $1-i$ in the configuration before, and this agent can interact with any agent in state $i$, so that $x_i(T')=x_i(T)=0$. If an interaction $(s_{BST}, i)\rightarrow (s'_{BST}, 1-i)$ occurs, then $x_i(T')=0$. In any case, the trace is such that $x_i(T')=0<x(T')$, so that any traces built upon $T'$ corresponds to some execution ending with some agent in state $1-i$ and one or zero agent in state $i$. Thus, any trace $T$ can be expanded with $(i, 1-i)\rightarrow(1-i, 1-i)$ interactions until at most one agent in state $i$ remains, and $x_i(T)\leq 1$.
\end{proof}
\noindent{\textcolor{darkgray}{$\blacktriangleright$}\nobreakspace\sffamily\bfseries  Lemma \ref{lemma-all-flip}.}
\emph{If phase $b$ starts with $n_b \geq n/2$, then it ends with $n_b = 0$
with probability at least $1/2$.}
\begin{proof}
For simplicity we will assume $b=0$; the $b=1$ case is symmetric.  So
we are looking at a zero phase that starts with $n_0 \geq n/2$.  From the
structure of the protocol, we know that at the start of this phase,
$c_1 = 0$, but $c_0$ might be larger.  It happens that the worst
case is when $c_0 = 0$, but we will analyze the process for any
initial value of $c_0$.

In the analysis below we will fix $n_0, n_1$,
to their values at the start of the phase.  To keep track of
what happens, let $i$ be the number of zero values converted to ones
so far during this phase; given the value of $i$, this gives $n_0-i$
zeros and $n_1+i$ ones in the population, and the value of the $c_1$
register will be $i$.  We fail to convert all zeros to ones if we exit
the phase while $i$ is less than $n_0$.

For each particular value of $i$, this occurs only if (a) $c_0$ is
already $0$, and (b) BST observes $6 (n \ln n + 1)$ ones
in a row.  Whether or not $c_0 = 0$, the latter event occurs with probability exactly
\begin{equation}
\left(\frac{n_1+i}{n}\right)^{6 (i \ln i + 1)}
\end{equation}
which by the union bound gives an upper bound on the probability
that we leave the phase for any $i < n_0$ of
\begin{equation}
\sum_{i=0}^{n_0-1} \left(\frac{n_1+i}{n}\right)^{6 (i \ln i + 1)}
\end{equation}

We will bound this sum by considering the terms with $i < n_0/2$ and
$i \geq n_0/2$ separately. 

For $i < n_0/2$, we have
\begin{align*}
    \sum_{i=0}^{\floor{(n_0-1)/2}} \left(\frac{n_1+i}{n}\right)^{6 (i \ln i + 1)}
    &\leq
    \sum_{i=0}^{\floor{(n_0-1)/2}} \left(\frac{n/2+n/4}{n}\right)^{6 (i \ln i + 1)}
    \\&=
    \sum_{i=0}^{\floor{(n_0-1)/2}} \left(3/4\right)^{6 (i \ln i + 1)}
    \\&\leq
    2\cdot(3/4)^6
    + \sum_{i=2}^{\infty} \left(3/4\right)^{6 (i \ln 2 + 1)}
    \\&=
    2\cdot(3/4)^6
    + (3/4)^6 \sum_{i=2}^{\infty} \left(3/4\right)^{6 i \ln 2}
    \\&=
    2\cdot(3/4)^6
    + (3/4)^6 (3/4)^{12 \ln 2} \frac{1}{1 - (3/4)^{6 \ln 2}}
    \\&\approx 0.37926
    \\&\leq 2/5.
\end{align*}

For $i \geq n_0/2$, we have
\begin{align*}
    \sum_{i=\ceil{n_0/2}}^{n_0-1} \left(\frac{n_1+i}{n}\right)^{6 (i \ln i + 1)}
    &\leq
    \sum_{i=\ceil{n_0/2}}^{n_0-1}
        \left(\frac{n_1+n_0-1}{n}\right)^{6 \left((n_0/2) \ln (n_0/2) + 1\right)}
    \\&\leq
        (n_0/2)
        \left(1-\frac{1}{n}\right)^{6 \left((n/4) \ln (n/4) + 1\right)}
    \\&\leq
        (n/4)
        \left(\exp\left(-\frac{1}{n}\right)\right)^{6 \left((n/4) \ln (n/4) + 1\right)}
    \\&=
        (n/4)
        \cdot
        e^{-(3/2) \ln (n/4) - 6}
    \\&=
        (n/4)
        \cdot
        n^{-3/2}
        \cdot
        e^{3/2 \ln 4 - 6}
    \\&\approx
        (0.004957)\cdot n^{-1/2}
    \\&\leq
        \frac{1}{200}\cdot n^{-1/2}.
    \\&\leq
        \frac{1}{200}.
\end{align*}

The original sum is thus bounded by $2/5+1/200 < 1/2$ for all
$n>0$, giving the claimed bound.
\end{proof}

\noindent{\textcolor{darkgray}{$\blacktriangleright$}\nobreakspace\sffamily\bfseries  Lemma \ref{lem-count-weak-property-sink1}.} \emph{Let $Count$ be a symmetric (silent or not) counting protocol correct under weak fairness (for any $\n\le P$). Consider any weakly fair execution $e = C_1, C_2, C_3, \ldots, C_j, \ldots$ of $Count$ on a population $\Agents$ of size $n< P$. There is an integer $k$ such that, for any $j\ge k$, no mobile agent is in a state $m\in Q$ such that there is a possible sequence of transitions of $Count$ $(m,m)\overset{*}{\rightarrow}(m,m)$.}
\begin{proof}
Let us assume, by contradiction, that there are infinitely many configurations in $e$ with a mobile agent in state $m$. Since there is a finite number of agents, there is a particular mobile agent $x$ in $\Agents$ which is in state $m$ in infinitely many configurations. Let $C_{j_1}, C_{j_2}, C_{j_3},\ldots$ be these configurations such that $e = e_1, C_{j_1}, e_2, C_{j_2}, e_3, C_{j_3},\ldots$ W.l.o.g., we choose these configurations such that, in every execution segment $e_i$, every agent in $\Agents$ interacts with every other  (this is possible with weak fairness).

Now consider a population $\Agents'=\Agents\cup \{x'\}$ of size $n+1$. To prove the lemma, we will construct a weakly fair execution $e'$ of $Count$ in population $\Agents'$ where no agent can distinguish $e'$ from $e$, and where consequently $Count$ wrongly counts only $n$ agents instead of the existing $n+1$.

We construct $e'$ based on $e$. 
First, we assume that in $e'$, $x'$ is in state $m$ in the starting configuration, and $e'=e'_1, C'_{j_1}, e^m, e'_2, C'_{j_2}, e^m, e'_3, C'_{j_3}, e^m,\ldots$ Every segment $e'_i$ follows exactly the same transition sequence as in $e_i$. In every segment $e'_{2r+1},C'_{j_{2r+1}}$ (for $r \ge 0$) the interactions are exactly the same as in $e_{2r+1},C_{j_{2r+1}}$, and $x'$ does not interact. However, in $e'_{2r},C'_{j_{2r}}$, all the interactions are as in $e_{2r+2},C_{j_{2r}}$, but the interactions with $x$. In this case, $x$ is replaced by $x'$ in the appropriate state, and $x$ does not interact. Finally, $e^m$ is an execution segment where only $x$ and $x'$ interact. They both start in state $m$, performing the sequence $(m,m)\overset{*}{\rightarrow}(m,m)$. The configurations at the beginning and at the end of $e^m$ are identical.  The construction of $e'$ ensures that in every $C'_{j_i}$, both $x$ and $x'$ are in the state $m$.

It is easy to verify that $e'$ is possible. In particular, this is because, at the end of every segment $e'_i, C'_{j_i}, e^m$, both $x$ and $x'$ are in the state $m$, so they can be exchanged in the following transitions of $e'_{i+1}$. Moreover, $e'$ is weakly fair, because $x'$ interacts with $x$ in every $e^m$; in every $e'_{2r+1}$ and $e'_{2r+2}$, $x$ and $x'$, respectively, interact with every other agent (by the assumption on $e_i$); and all the other agents interact with all the others infinitely often, by the later arguments and by weak fairness of $e$.

Finally, in $e'$, every agent from $\Agents$ (including BST), executes exactly the same sequence of transition rules as it does in $e$, so no agent can distinguish the fact that the population is actually $\Agents'$ with $n+1$ agents, and $Count$ counts only $n$ agents as it does in $e$. This is a contradiction to the assumption that $Count$ is a correct counting protocol.
\end{proof}

\noindent{\textcolor{darkgray}{$\blacktriangleright$}\nobreakspace\sffamily\bfseries  Proposition \ref{prop-count-weak-Ps}.} \emph{Any symmetric counting protocol correct for any $\n\le P$ (undistinguishable and non-initialized) mobile agents, under weak fairness, have to use at least $P$ states per mobile agent.}
\begin{proof}
As $Count$ is symmetric, any two agents, both in some state $s\in Q$, in an interaction, have to execute a symmetric transition of the form $(s,s)\rightarrow(s_1,s_1)$. Thus there is a possible sequence of transitions $(s,s)\rightarrow(s_1,s_1)\rightarrow(s_2,s_2)\rightarrow(s_3,s_3)\ldots$ As mobile agents are finite state, for some $j> i\ge 1$, $s_{i}=s_{j}$, i.e. $(s_i, s_i)\overset{*}{\rightarrow}(s_i,s_i)$.  By Lem. \ref{lem-count-weak-property-sink1}, $s_i$ does not appear infinitely often in executions with $\n< P$. As there are at least $P-1$ states appearing infinitely often in an execution with $\n= P-1$ (by Proposition \ref{prop-count-weak-property-naming}), and there is at least one such state $s_i$, there are at least $P$ distinct states to be maintained by mobile agents.
\end{proof}

\begin{proposition}\label{prop-2P:time}
The convergence time of the space optimal protocol under weak fairness presented in \cite{DBLP:conf/wdag/BeauquierBCS15}  is $\Theta(2^\n)$ non-null transitions. 
\end{proposition}
\begin{proof}
Following the study of the Gros sequence $\U_{n}$, in \cite{DBLP:conf/wdag/BeauquierBCS15},
used by the protocol to name $n$ agents ($\U_n\equiv\U_{n-1}, n, \U_{n-1}$, where $\U_1\equiv 1$), the number of terms in $\U_{n}$ is $2^n-1$. In consequence, the number of non-null transitions, before convergence, between BST and an agent in state $0$ (the $m$-state) or $>n$ (whom BST ``gives'' a new name), is at most $2^{\n}$.
Other possible non-null transitions are between homonyms. Each agent in such a transition changes its state to 0, and its next non-null transition necessarily involves BST. 
Thus, there cannot be more than $2^{\n}$ non-null transition between homonyms. Hence, 
the proposition follows.
%
\end{proof}

\section*{Convergence Time Analysis of Protocol \ref{alg:counting2s} \cite{DBLP:conf/wdag/BeauquierBCS15}} 

\begin{algorithm}[h]
\caption{-- Space Optimal Counting under Global or Probabilistic Fairness \cite{DBLP:conf/wdag/BeauquierBCS15}}
\begin{algorithmic}
    \STATE
	\STATE \textbf{Variables at BST:}
    \INDSTATE $\size_0$: non-negative integer, initialized to $0$; 
    \INDSTATE $\size_1$: non-negative integer, initialized to $0$; 
	\INDSTATE $\sizeT$: non-negative integer initialized to $0$; eventually holds $\n$
	\STATE \textbf{Variable at a mobile agent $x$:}
	\INDSTATE $b$: in $\{0,1\}$, initialized \emph{arbitrarily}
	\STATE
\end{algorithmic}

\begin{algorithmic}[1]

\WHEN{\underline{a mobile agent $x$ interacts with BST}}
	\IF{$\size_{b}>0$}
		\STATE $\size_{b} \leftarrow \size_{b}-1$\label{l-2s:size-mark--}
	\ENDIF
	\STATE $b \leftarrow 1 - b$ \label{l-2s:change-mark}
    \STATE $\size_{b} \leftarrow \size_{b}+1$ \label{l-2s:size-mark++}
    \STATE $\sizeT\leftarrow \size_{0}+\size_{1}$\label{l-2s:size-update}
\ENDWHEN
\end{algorithmic}\label{alg:counting2s}
\end{algorithm}
%
%
We evaluate the convergence time of Protocol \ref{alg:counting2s} in terms of the average number of transitions, assuming probabilistic fairness (uniformly random interactions). Note that, in case of this protocol, any transition involves BST. 
Below, we sketch the analysis and after we give the details. 
%


To compute the convergence time, we use the observation (stated by Lemma \ref{lm-2s:for-time} appeared and proven below) 
that Protocol \ref{alg:counting2s} must first reach a configuration with all mobile agents in the same state, and then a configuration with all the agents in the other state (recall that there are only two mobile agent states).

Thus, consider a population of $\n$ agents, and let $t_k$ be the average number of transitions that happen before all agents are in state 0, starting from a configuration with $k$ agents in state 1. Then, $t_0=0$ trivially. For $1\leq k\leq \n-1$, $t_k=1+\frac kn t_{k-1}+\frac{\n-k}\n t_{k+1}$. This is because, at the current step, there are $k$ chances out of $\n$ that an agent in state 1 meets BST, leading to a configuration with $k-1$ agents in state 1; and $\n-k$ chances out of $\n$ that an agent in state 0 interacts with BST resulting in a configuration with $k+1$ agents in state 1. Finally, $t_\n=1+t_{\n-1}$.

From that we have $t_\n=2^{\n-1}\sum_{k=0}^{\n-1}\frac{1}{{\n-1\choose k}}=2^\n+o(2^\n)$, 
and the average convergence time of Protocol \ref{alg:counting2s} is $\Theta(2^\n)$ transitions. The formal proof is below.

\begin{lemma}
\label{lm-2s:for-time}
In the first configuration $C^1$ after the convergence of Protocol \ref{alg:counting2s}, i.e., the first time when $\sizeT=\n$ (and does not change after), all agents have the same mark $m \in \{0,1\}$. Moreover, there is a configuration $C^0$ s.t. $C^0\overset{*}{\rightarrow} C^1$, and all agents in $C^0$ are in state $1-m$.
\end{lemma}
\begin{proof}
Thus, the counting is achieved in $C^1$. This happens following an interaction of BST with
an agent, let us say w.l.o.g., in state 0. By definition, $\sizeT=\size_0+\size_1=\n$, and since $\sizeT$ increases in this interaction, we had and we still have $\size_0=0$ (otherwise, $\sizeT$ cannot increase). Then, after this transition, $\size_1$ becomes $\n$.

We show now that at the last transition with $\size_1=0$, before $C^1$ has been reached (at least the first step is such), all agents were in state 0 (this will prove the existence of $C^0$). Denote by $r$ the number of 1-0 transitions (transitions changing a state of a mobile agent from 1 to 0). Then, the number of 0-1 transitions is $\n+r$, since 1-0 transitions increment $\sizeT$, and 0-1 ones decrement it ($\size_1$ never goes to 0, by assumption). Thus, BST meets $\n+r$ agents in state 0 that turn to 1, and $r$ in state 1 that turn to 0. This creates $\n$ new agents in state 1 thus, at the step when $\size_1=0$, all agents were in state 0.
\end{proof}

Thus, consider a population of $\n$ agents, and let $u_k$ be the average number of transitions that happen before all agents are in state 0, starting from a configuration with $k$ agents in state 1.

We have the following relations:\begin{itemize}
\item $u_0=0$ by definition;
\item for $1\leq k\leq \n-1$, $u_k=1+\frac kn u_{k-1}+\frac{\n-k}\n u_{k+1}$: at the current step, there is $k$ chances in $\n$ that an agent with mark 1 meets BST, leading to a configuration with $k-1$ agents with mark 1, and $\n-k$ chances in $\n$ that an agent marked 0 interacts with BST resulting in a configuration with $k+1$ agents marked 1;
\item $u_\n=1+u_{\n-1}$.
\end{itemize}

For $0\leq k\leq \n-1$, set $v_k=u_{k+1}-u_k$. We have:\begin{itemize}
\item $v_{\n-1}=1$
\item $\forall 1\leq k\leq \n-1, v_{k-1}=u_{k}-u_{k-1}=u_k-\frac{\n}{k}\left(u_{k}-1-\frac{\n-k}\n u_{k+1}\right)=\frac{k-\n}{k}u_{k}+\frac{\n}{k}+\frac{\n-k}{k}u_{k+1}=\frac{\n-k}{k}v_k+\frac{\n}{k}$
\end{itemize}

Thus, for $0\leq k\leq \n-1$

$$v_{\n-k}=\prod_{i=\n-k+1}^{\n-1}\frac{\n-i}i+\sum_{i=1}^{k-1}\prod_{j=i}^{k}\frac{j}{\n-j}\frac{\n}{i}=\frac{(\n-k)!(k-1)!}{(\n-1)!}+\n\sum_{i=1}^{k-1}\frac{(k-1)!(\n-k)!}{i!(\n-i)!}=$$
$$=\n\sum_{i=0}^{k-1}\frac{(k-1)!(\n-k)!}{i!(\n-i)!}$$

$$v_{\n-k}=\frac1{{\n-1\choose k-1}}+\n\sum_{i=1}^{k-1}\frac{(k-1)!(\n-k)!}{\n!}\frac{\n!}{i!(\n-i)!}=\frac1{{\n-1\choose k-1}}+\sum_{i=1}^{k-1}\frac{{\n\choose i}}{{\n-1\choose k-1}}=\frac{\sum_{i=0}^{k-1}{\n\choose i}}{{\n-1\choose k-1}}$$

From that, we get:

$$u_\n=\sum_{k=0}^{\n-1}v_k=\sum_{k=0}^{\n-1}\frac{\sum_{i=0}^{k}{\n\choose i}}{{\n-1\choose k}}$$

First, consider the case when $\n$ is even:
$$u_\n=\sum_{k=0}^{\n/2-1}\frac{\sum_{i=0}^{k}{\n\choose i}}{{\n-1\choose k}}+\sum_{k=\n/2}^{\n-1}\frac{\sum_{i=0}^{k}{\n\choose i}}{{\n-1\choose k}}$$

$$u_\n=\sum_{k=0}^{\n/2-1}\frac{\sum_{i=0}^{k}{\n\choose i}}{{\n-1\choose k}}+\sum_{k=\n/2}^{\n-1}\frac{2^\n-\sum_{i=k+1}^{\n}{\n\choose i}}{{\n-1\choose k}}$$
since $\sum_{i=0}^{\n}{\n\choose i}=2^\n$
$$u_\n=\sum_{k=0}^{\n/2-1}\frac{\sum_{i=0}^{k}{\n\choose i}}{{\n-1\choose k}}+\sum_{k=\n/2}^{\n-1}\frac{2^\n-\sum_{i=0}^{\n-k-1}{\n\choose \n-i}}{{\n-1\choose \n-k-1}}$$
by setting $i'=\n-i$
$$u_\n=\sum_{k=0}^{\n/2-1}\frac{\sum_{i=0}^{k}{\n\choose i}}{{\n-1\choose k}}+\sum_{k=\n/2}^{\n-1}\frac{2^\n-\sum_{i=0}^{\n-k-1}{\n\choose i}}{{\n-1\choose \n-k-1}}$$
$$u_\n=\sum_{k=0}^{\n/2-1}\frac{\sum_{i=0}^{k}{\n\choose i}}{{\n-1\choose k}}+\sum_{k=0}^{\n/2-1}\frac{2^\n-\sum_{i=0}^{k}{\n\choose i}}{{\n-1\choose k}}$$
by setting $k'=\n-k-1$
$$u_\n=2^\n\times\sum_{k=0}^{\n/2-1}\frac{1}{{\n-1\choose k}}$$
$$u_\n=2^{\n-1}\times\sum_{k=0}^{\n-1}\frac{1}{{\n-1\choose k}}$$

The case when $\n$ is odd is similar:
$$u_\n=\sum_{k=0}^{(\n-3)/2}\frac{\sum_{i=0}^{k}{\n\choose i}}{{\n-1\choose k}}+\frac{\sum_{i=0}^{(\n-1)/2}{\n\choose i}}{{\n-1\choose (\n-1)/2}}+\sum_{k=(\n+1)/2}^{\n-1}\frac{\sum_{i=0}^{k}{\n\choose i}}{{\n-1\choose k}}$$
And, similarly
$$u_\n=\sum_{k=0}^{(\n-3)/2}\frac{\sum_{i=0}^{k}{\n\choose i}}{{\n-1\choose k}}+\sum_{k=0}^{(\n-3)/2}\frac{2^\n-\sum_{i=0}^{k}{\n\choose i}}{{\n-1\choose k}}+\frac{\sum_{i=0}^{(\n-1)/2}{\n\choose i}}{{\n-1\choose (\n-1)/2}}$$
by setting $k'=\n-k-1$
$$u_\n=2^\n\times\sum_{k=0}^{(\n-3)/2}\frac{1}{{\n-1\choose k}}+2^{\n-1}\times\frac{1}{{\n-1\choose (\n-1)/2}}$$
$$u_\n=2^{\n-1}\times\sum_{k=0}^{\n-1}\frac{1}{{\n-1\choose k}}$$

Now, for $2\leq k\leq \n-3$, $\frac{1}{{\n-1\choose k}}\leq \frac{1}{{\n-1\choose 2}}=\frac{2}{(\n-1)(\n-2)}$, so that
$$2=\frac{1}{{\n-1\choose 0}}+\frac{1}{{\n-1\choose \n-1}}\leq\sum_{k=0}^{\n-1}\frac{1}{{\n-1\choose k}}\leq 2+\frac{2}{\n-1}+\frac{2(\n-4)}{(\n-1)(\n-2)}=2+O\left(\frac1\n\right)$$

Thus, $u_\n\geq 2^\n$, and $u_\n\sim_{\n\rightarrow\infty}2^\n$.

The average complexity of the protocol is $\Theta(2^\n)$. The best starting configurations, for the complexity, is when all agents have the same mark. The average complexity is then $2^\n+o(2^\n)$. Starting from any other configuration, the protocol first has to reach a configuration where all agents have identical marks. This takes less than $2^\n+o(2^\n)$ transitions, since starting with all agents having the same mark, and switching it, makes the protocol traverse all configurations. Hence, in this case, the overall complexity is less than $2\times (2^\n+o(2^\n))$.

\end{document}